\documentclass[runningheads]{llncs}
\usepackage[T1]{fontenc}
\usepackage{booktabs}
\usepackage{tabularx}
\usepackage{hyperref}
\usepackage{csquotes}
\usepackage{amsmath}
\usepackage{cleveref}
\usepackage{amsfonts}
\usepackage{siunitx}
\usepackage{algorithm}
\usepackage{tikz}
\usetikzlibrary{patterns}
\usepackage{graphicx}
\usepackage{fontawesome5}
\usepackage{xcolor}

\newcommand*\prot{OmniSphinx}
\newcommand*\secpar{\kappa}
\newcommand*\clients{\mathcal{C}}
\newcommand*\mixnodes{\mathcal{M}}
\newcommand*\corrupted{\check{\mixnodes}}

\newcommand*\instr[1]{\texttt{#1}}
\newcommand*\byte{\mathbb{N}_{255}}
\newcommand*\bytes{\mathbb{N}_{255}^*}
\newcommand*\macSet{\mathbb{N}_{255}^\secpar}
\newcommand*\regSet{\mathcal{R}}
\newcommand*\regCnt{|\regSet|}

\newcommand*\header{\eta}
\newcommand*\payload{\delta}
\newcommand*\group{\mathcal{G}}
\newcommand*\Zg{\mathbb{N}_q}
\newcommand*\mixprog{p}
\newcommand*\pubkey{y}
\newcommand*\unidraw{\leftarrow^\$}
\newcommand*\concat{||}
\newcommand*\abs[1]{\lvert #1 \rvert}
\newcommand*\substr[2]{_{[#1 .. #2]}}
\newcommand*\xor{\oplus}
\newcommand*\pz{\phantom{0}}

\newcommand*\bad{\textit{malignant}}
\newcommand*\good{\textit{benign}}

\newcommand*\inp[1]{\textbf{Input:} #1}
\newcommand*\cmt[1]{\textcolor{gray}{\textit{\texttt{// #1}}}}

\newif\ifextended

\begin{document}

\extendedtrue

\ifextended
    \title{OmniSphinx: Active Mix Networks}
    \subtitle{Extended Version}
\else
    \title{OmniSphinx: Active Mix Networks}
\fi
\author{
    Daniel Schadt\inst{1}\textsuperscript{*}\orcidID{0009-0009-6357-1314}
    \and
    Christoph Coijanovic\inst{1}\orcidID{0000-0002-5873-2859}
    \and
    Shabi Shabani\inst{2}\orcidID{0009-0001-7308-7587}
    \and
    Thorsten Strufe\inst{1}\orcidID{0000-0002-8723-9692}
}
\authorrunning{D. Schadt et al.}
\institute{
    Karlsruhe Institute of Technology, Karlsruhe, Germany\\
    \email{\{daniel.schadt,christoph.coijanovic,thorsten.strufe\}@kit.edu}
    \and
    Karlsruhe Institute of Technology, Karlsruhe, Germany\\
    \email{mail@shabi-shabani.de}
}
\maketitle              % typeset the header of the contribution
\begin{abstract}
Mix networks are an important tool to implement anonymous communication, which protects not just the content but also the metadata of messages.
Over time, various packet formats for mix networks have been proposed, usually with single, specific goals in mind.
These formats are incompatible with each other, requiring separate software and infrastructure to be set up.
In this paper, we propose a new format, \emph{\prot{}}, which solves this issue.
In \prot{}, senders embed code in their packets that determines how they must be processed.
The resulting active mix network can emulate any other mix format within a single deployment.
Our empirical evaluation shows that emulation in \prot{} incurs reasonable overhead compared to native execution for typical mix network use cases:
For Sphinx, the most compact format, computation time increases by around \qty{90}{\us}, while headers increase by 33\% in size.

\end{abstract}

\makeatletter
\protected@xdef\@thefnmark{*}
\@footnotetext{Corresponding author}
\makeatother
\section{Introduction}
% Broad Topic
The internet has become an invaluable tool for modern communication.
Such communication, however, can  be observed by the service providers and third parties.
Many services thus employ encryption to protect the content of their communication, but metadata -- information such as \emph{who talks to whom} or \emph{when do they talk} -- is left unprotected.
This metadata alone often reveals sensitive information, for instance a frequent contact with medical professionals or specific political groups.

% Specific Topic

Mix networks hide such metadata by relaying a message over a series of intermediaries (the \emph{mix nodes}).
This provides anonymity, but comes at the cost of fixed message layouts and limited functionalities:
As each packet must be encrypted in a specific format, only communication patterns that are supported by the chosen format can be used.
This puts pressure on mix network operators to support the most popular format, and limits users that need specific functionalities like multicast traffic.
Ideally, a single mix network would not just support any format that currently exists~\cite{danezis2009,hugenroth2021,schadt2024,cryptoeprint:2024/020,rial2025outfox}, but also be easily extended for future formats without requiring coordinated infrastructure updates.

% Background
To work towards such a mix network, we take ideas from \enquote{active networking}~\cite{tennenhouse1996towards,bhattacharjee1997architecture,alexander2002switchware}, a line of research that aimed to add customization to general purpose network nodes (e.g., to implement custom compression).
A promising concept is to combine the code needed to process a packet with the payload in a self-contained structure called a \emph{capsule}.
Due to the performance penalty that this approach induces, it has not been adopted widely, especially as use-cases were lacking to justify the cost~\cite{feamster2014road}.
However, we see a good use-case in mix networks because it allows the nodes to be oblivious to the different mix formats.

% Goal
Capsules inherently cause overhead, so the goal of this work is to investigate whether this approach is still viable within the scope of mix networks:
As message decryption and shuffling already incur a higher latency than direct communication, we expect the performance penalty to be less of an issue.
At the same time, there are concrete benefits to adding such dynamicity to mix networks:
A single mix network instance can be used by clients with different requirements regarding the mix format, leading to better utilization of the instance, and potentially, a larger anonymity set.
Further, if more mix nodes provide the functionality a client needs, they can choose from a more diverse set of nodes and operators.

% Requirements & metrics to evaluate
We identify three challenges when applying active networking to mix formats:
First, we have to ensure that the resulting format is flexible enough to emulate all relevant existing mix formats.
Second, this flexibility should not come at the cost of privacy --
an emulated mix format must achieve the same privacy guarantees as the native one.
Third, the added flexibility should not incur costs that inhibit typical use cases of mix networks, such as email communication.

% Overview of solution
For our investigation, we develop \prot{}, a novel \emph{active} mix format that achieves flexibility in mix networks.
\prot{} is based on Sphinx, but extends the packet header with a \emph{mix program} for each node on the packet's path.
These mix programs are built from an instruction set that is tailored towards the required operations for existing mix formats, and they determine how the packet's payload should be processed at each node.
This allows \prot{} to emulate all relevant existing mix formats.
Our evaluation shows that \prot{} can emulate Sphinx with an overhead of 33\% in header size, and additional \qty{90}{\us} in computation time.

Finally, we recognize that \prot{} requires senders to be careful when assembling mix programs, as to not make packets traceable by accident.
To help senders in constructing \enquote{secure} mix programs, we provide \emph{information flow} rules that can be used to assess whether a mix program leaks information.

% Contributions
In summary, we make the following main contributions:
\begin{itemize}
    \item We develop \prot{}, a novel mix format that implements ideas from active networking to make mix networks flexible.
    \item We present an instruction set that allows \prot{} to emulate all existing relevant mix formats.
    \item We introduce information flow analysis to the context of anonymous communication to argue about the privacy protection of arbitrary mix programs.
    \item We provide an empirical evaluation to determine bandwidth and computational overhead of \prot{} compared to non-flexible mix formats.
\end{itemize}

This paper is structured as follows:
\Cref{sec:background} introduces the necessary background.
\Cref{sec:related} compares our goals to related work.
In \Cref{sec:model}, we introduce our assumed system and threat model.
\Cref{sec:design} gives a high-level overview of \prot{}'s design, while \Cref{sec:format} discusses the format in more technical detail.
In \Cref{sec:programs}, we describe how \prot{} can emulate Sphinx and PolySphinx.
\Cref{sec:security} discusses \prot{}'s security guarantees, while \Cref{sec:evaluation} discusses its performance impact.
Finally, \Cref{sec:conclusion} concludes this work.

\section{Background}
\label{sec:background}
In this section, we introduce the general notation as well as specific background knowledge required for this paper.

\subsection{Notation}

We denote the security parameter as \(\secpar\).
Breaking the cryptographic primitives requires work that scales exponentially in \(\secpar\).
We denote the set of possible bytes as \(\byte\), the set of byte strings of length \(l\) as \(\byte^l\), and the set of byte strings of arbitrary length as \(\bytes\).
Given \(a \in \byte^x\) and \(b \in \byte^y\), \(a \concat b \in \byte^{x+y}\) denotes the concatenation of \(a\) followed by \(b\).
If \(x = y\), \(a \xor b \in \byte^x\) denotes the bitwise xor of \(a\) and \(b\).
Further, \(a\substr{i}{j}\) with \(i \le j < x\) denotes the substring of \(a\) from \(i\) to \(j\), and \(0^i\) denotes the string of \(i\) zeroes.
We denote a Diffie-Hellman group as \(\group\).
Its order is \(q\) and its generator \(g\).

\subsection{Mix networks}

Mix networks~\cite{chaum1981} are a tool for anonymous communication.
In a mix network, Alice sends a message to Bob over a series of intermediaries (the \emph{mix nodes}), which hides their relation.
To prevent an outside observer from tracking messages as they pass through mix nodes, two techniques are employed:

First, messages are encrypted multiple times, with each mix node removing one layer of encryption.
This concept is colloquially known as onion encryption and prevents the observer from tracking messages based on their content.
A \emph{mix format} describes the exact way how messages are encrypted.

Second, messages are intentionally delayed and shuffled, either by waiting for a random time before they are forwarded, or by waiting for a certain amount of other messages.
This prevents the observer from tracking messages based on their timing or output order.
A \emph{mixing strategy} describes the exact way how messages are delayed and shuffled.

A mix network provides anonymity to Alice even against global active adversaries, as long as any one mix node in Alice's chosen path is honest.
This is called the \emph{anytrust} assumption.

\subsection{Mix formats}

The mix format describes how a sender has to prepare a message, and how a mix node has to process it.
Sphinx~\cite{danezis2009} is a prominent example of a mix format, which promises compactness and provable security.

A Sphinx packet consists of a header and the payload.
The header contains a key encapsulation \(g^x\) based on a Diffie-Hellman key exchange, which the node uses to derive a shared secret \(g^{xy}\) using its private key \(y\).
This shared secret keys a stream-cipher, which is used to decrypt the remaining header, as well as a MAC, which is used to verify the integrity of the header, as well as a block-cipher, which is used to decrypt the payload.
After decryption, the node can read the address of the next hop.
Finally, the shared secret is blinded to produce a new key encapsulation, allowing the next node to derive a fresh shared secret.

Thanks to its compactness, Sphinx has been used as the basis for later adaptions and extensions:
Beato et al. propose an improvement of Sphinx that uses authenticated encryption instead of the explicit hashes that Sphinx uses~\cite{beato2016}.
We call this variant \emph{AE-Sphinx}.
Hugenroth et al. propose \emph{MultiSphinx}~\cite{hugenroth2021}, a format that allows multiple (small) Sphinx packets to be combined into a single big packet.
This allows for a more efficient sending of multiple small packets.
Schadt et al. propose \emph{PolySphinx}~\cite{schadt2024}, a format that allows for efficient replication of messages with the same content to multiple recipients.
This allows for efficient group communication.

Outside of Sphinx, Klooß et al. propose \emph{EROR}~\cite{cryptoeprint:2024/020}, which requires weaker trust assumptions than Sphinx.
EROR includes hop-by-hop integrity protection, which prevents tagging attacks in the presence of malicious receivers.

\section{Related work}
\label{sec:related}
The goal of this work is to define a packet format for mix networks that offers flexibility to emulate existing and future mix formats.

\prot{} shares its setting and basic architecture with Sphinx~\cite{danezis2009} and Sphinx-like mix formats~\cite{hugenroth2021,schadt2024}.
Unlike \prot{}, these formats rigidly define packet processing as part of the protocol and do not offer flexibility.

Rochet et al. propose \textbf{FAN}~\cite{DBLP:conf/wpes/RochetDE24}, a protocol that enables dynamic reprogramming of Tor~\cite{DBLP:conf/uss/DingledineMS04} nodes.
FAN's goal is to increase the speed with which updates to the Tor code base propagate through the network.
FAN offers flexibility in the sense that the protocol can be easily updated.
However, as packet processing is identical for all packets and determined by the canonical code base, it does not meet our notion of flexibility.

With \textbf{Bento}~\cite{DBLP:conf/ccs/ReiningerAHFGL20,DBLP:conf/sigcomm/ReiningerAHFHGL21}, Reininger et al. extend Tor with network function virtualization.
Bento adds additional nodes to the Tor network to which users can upload code that is executed within trusted execution environments.
This can be used to augment Tor's functionality, for example with cover traffic generation or load balancing.
It does however have no impact on the actual packet processing, as it leaves the Tor nodes unchanged.

In the context of censorship circumvention, protocols like \textbf{Marionette}~\cite{dyer2015} and \textbf{Proteus}~\cite{wails2023proteus} allow proxies to be re-programmed quickly by clients.
However, the focus of those systems is to be flexible in order to evade censorship systems.
They do not process messages but instead aim to tunnel TCP streams covertly, and they lack the privacy guarantees that mix networks provide.

\section{System and threat model}
\label{sec:model}
We assume a mix network consisting of \emph{clients}, which want to send messages to each other, and \emph{mix nodes}, which are the servers providing the service.
We assume a service model~\cite{kuhn2020} for the network, meaning that the final mix node in a path forwards a message to the recipient using a non-anonymous channel.
We call the set of client addresses \(\clients\) (e.g., the set of fixed-length email addresses) and the set of mix node addresses \(\mixnodes\) (e.g., IP addresses).

We assume that a public key infrastructure exists such that each mix node \(m \in \mixnodes\) has a private key \(x_m\), which is kept secret, and a public key \(y_m = g^{x_m}\) that is known by the clients. 

We assume that the adversary is global and active, meaning that it can observe any network link between two participants, and it can insert, modify, delay and drop messages on those links.
We further assume that the adversary can control a subset \(\corrupted \subset \mixnodes\).
We call nodes in \(\corrupted\) \emph{corrupt}, and the remaining nodes \emph{honest}.
The adversary is bounded to use polynomial time in \(\secpar\).
The adversary's goal is to link sender and recipient of messages sent.
We do not aim to hide the processing logic from the adversary or from the mix nodes.

\section{Design}
\label{sec:design}
We introduce \emph{\prot{}}, a new mix format that lets senders specify how each individual packet should be processed.
In \prot{}, senders embed a \emph{mix program} for every node into their packet, and nodes execute the mix program to process the packet.
By using different mix programs, senders can \emph{emulate} existing and future mix formats.
Note that the resulting \prot{} packets are not byte-wise compatible to the emulated mix format; rather, they implement the same conceptual design, meaning they provide the same functionality and achieve the same privacy guarantees.

\prot{} is based on Sphinx~\cite{danezis2009}.
It structures a packet into two parts, the header and the payload.
The header contains the mix programs, whereas the payload contains the sender's message.
Like Sphinx, the base packet format employs onion encryption and MACs in the header to ensure integrity of the mix programs and that each mix node only learns \enquote{their} mix program.
Unlike existing mix formats, \prot{} does not specify how payloads should be handled to give senders as much flexibility as possible.
All payload processing is done according to the mix program.

Each mix program consists of a sequence of \emph{instructions}.
The available instructions are defined by the \emph{instruction set}.
The instruction set is chosen to provide a good trade-off between flexibility and overhead:
If the instructions are too low-level (e.g., x86 machine code), representing a mix program takes up many instructions and becomes inefficient.
On the other hand, if the instructions are too high-level (e.g., switching between different existing formats), then \prot{} deployments will not support new formats without updating the instruction set, which defeats its purpose.
Additionally, the instructions are chosen to limit attacks from malicious users against mix nodes, like the attempted extraction of secrets or the abuse of computation time.

Packet processing at mix nodes is divided into three stages:
\begin{itemize}
    \item During \emph{preprocessing}, the node derives the shared secret and unwraps the onion encryption in the header.
    This reveals the mix program for the current node.

    \item During \emph{program execution}, the node executes the instructions specified in the mix program.
    The program has access to the header, payload, and shared secret.
    A dedicated \instr{Forward} instruction is used to enqueue outgoing packets for transmission.

    \item During \emph{postprocessing}, the mix node ensures that each outgoing packet has the correct size, padding it if necessary.
\end{itemize}

As outgoing packets are created by the mix program, a single packet can result in no outgoing packets (e.g., to implement cover traffic), a single outgoing packet (e.g., a normal message), or multiple outgoing packets (e.g., for multicast traffic).

\section{Implementation}
\label{sec:format}
In this section, we define the \prot{} format in detail.

\subsection{Cryptographic primitives}

We make use of a number of cryptographic primitives in our implementation of \prot{}:
We use a secure pseudorandom generator, \(\rho : \byte^\secpar \rightarrow \bytes\), which takes as input a seed and outputs arbitrarily many random looking bytes.
We use a MAC, which we model as a random oracle \(\mu : \byte^\secpar \times \bytes \rightarrow \byte^\secpar\), taking the input key and bytes and producing an \enquote{authentication tag.}
We use hash functions \(h_b : \group \times \group \rightarrow \Zg\) to blind group elements between hops, \(h_\rho, h_\mu : \group \rightarrow \byte^\secpar\) and \(h_\rho' : \group \times \byte \times \mathbb{N} \rightarrow \byte^\secpar\) to derive keys from group elements, and \(h_\tau : \group \rightarrow \byte^\secpar\) to identify already-seen messages.

\subsection{Instruction architecture}
\label{sec:format:instructions}

Instructions in \prot{} follow a register-based approach, where intermediate values are saved in one of \(\regCnt = 256\) registers.
Each register \(r \in \regSet\) can hold a byte string \(b \in \bytes\).
Each instruction takes a predefined number of arguments, where each argument can be either the index of a register to use as in-/output, or a constant value.

We write instructions as \instr{Name(a, b, ...)}, where \instr{Name} is the name of the instruction and \(\instr{a}, \instr{b}, \ldots \in \regSet \cup \byte\) are its arguments.
Internally, each instruction is encoded as a byte representing the operation code, followed by bytes representing the arguments.

A special case is the \instr{Load} instruction, which is used to embed byte string constants into the mix program.
The \instr{Load} instruction takes the form \instr{Load(v, d)}, where \instr{v} is the value to load and \(\instr{d} \in \regSet\) is the destination register.
Internally, \instr{Load} is encoded with a byte representing the operation code for \instr{Load}, followed by two bytes specifying the length of \instr{v}, followed by the bytes of \instr{v}, followed by a byte representing \instr{d}.

Our instruction set is based on an analysis of Sphinx and its derivatives and it is chosen so that those formats can be emulated.
A full list of the available instructions can be found in \Cref{tab:instructions}.

\begin{table}[]
    \centering
    \begin{tabularx}{\textwidth}{lX}\toprule
         Instruction & Description\\\midrule
         \(\instr{Exponent}(b,e,d)\) & Computes \(b^e\) and stores the result in \(d\)\\
         \(\instr{ConcatByte}(a,b,d)\) & Concatenates the byte value of \(b\) to \(a\) and stores the result in \(d\)\\
         \(\instr{Concat}(a,b,d)\) & Concatenates the content of \(b\) to \(a\) and stores the result in \(d\)\\
         \(\instr{IsEqual}(a,b)\) & Checks if \(a\) and \(b\) are equal. Aborts if false.\\
         \(\instr{CutBytes}(a,l,d)\) & Extracts the first \(l\) bytes of \(a\) and stores the result in \(d\)\\
         \(\instr{XOR}(a,b,d)\) & Computes the bitwise XOR of \(a\) and \(b\) and stores the result in \(d\)\\
         \(\instr{Add}(a,b,d)\) & Computes the sum of \(a\) and \(b\) and stores the result in \(d\)\\
         \(\instr{Copy}(s,d)\) & Copies the content of \(s\) to \(d\)\\
         \(\instr{Pad}(a,l,d)\) & Appends \(l\) \enquote{0}s to \(a\) and stores the result in \(d\)\\
         \(\instr{Load}(s,d)\) & Loads the constant or intermediate value \(s\) into \(d\)\\
         \(\instr{CreateZeroes}(l,d)\) & Stores \(l\) \enquote{0}s in \(d\)\\\addlinespace
         \(\instr{PRG}(s,l,d)\) & Expands seed \(s\) into a pseudorandom string of length \(l\) and stores the result in \(d\)\\
         \(\instr{Hash}(s,d)\) & Applies the hash function to \(s\) and stores the result in \(d\)\\
         \(\instr{Encrypt}(a,b,d)\) & Encrypts plaintext \(b\) under key \(a\) and stores the result in \(d\)\\
         \(\instr{Decrypt}(a,b,d)\) & Decrypts ciphertext \(b\) under key \(a\) and stores the result in \(d\)\\
         \(\instr{MAC}(a,b,d)\) & Computes a MAC of \(b\) using key \(a\) and stores the result in \(d\)\\\addlinespace
         \(\instr{ForLoop}(a,b)\) & Repeat the next \(a\) instructions \(b\) times\\
         \(\instr{Forward}(a)\) & Forward the computed data to \(a\)\\
         \(\instr{Stop}()\) & Terminate the program execution\\\bottomrule
    \end{tabularx}
    \caption{The \prot{} instruction set.}
    \label{tab:instructions}
\end{table}

\subsection{Packet structure}
\label{sec:format:packet}

We consider a packet to consist of two parts:
A \emph{header} \(\header\) and a \emph{payload} \(\payload\).
\prot{} makes no assumptions about the payload and treats it as an opaque block of bytes \(\delta \in \byte^l\).
The header has a pre-defined structure and contains all information required to process the packet at a mix node.

An \prot{} header consists of three parts: \(\header = (\alpha, \beta, \gamma)\).
Here, \(\alpha \in \group\) is a Diffie-Hellman group element, \(\beta \in \byte^j\) contains the mix programs and message authentication tags for the nodes along the path, and \(\gamma \in \macSet\) is the message authentication tag for the first node.
We show this structure in \Cref{fig:packet-structure}.

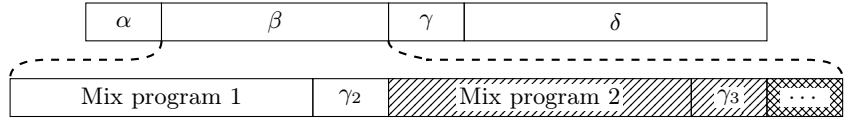
\begin{figure}
    \centering
    \begin{tikzpicture}
        \draw (0, 0) rectangle ++(1, 0.5) node[midway] { \(\alpha\) };
        \draw (1, 0) rectangle ++(3, 0.5) node[midway] { \(\beta\) };
        \draw (4, 0) rectangle ++(1, 0.5) node[midway] { \(\gamma\) };
        \draw (5, 0) rectangle ++(4, 0.5) node[midway] { \(\delta\) };

        \draw (-1, -1) rectangle ++(4, 0.5) node[midway] { Mix program 1 };
        \draw (3, -1) rectangle ++(1, 0.5) node[midway] { \(\gamma_2\) };
        \draw[pattern=north east lines] (4, -1) rectangle ++(4, 0.5) node[midway,fill=white, inner sep=1pt] { Mix program 2 };
        \draw[pattern=north east lines] (8, -1) rectangle ++(1, 0.5) node[midway,fill=white,inner sep=1pt] { \(\gamma_3\) };
        \draw[pattern=crosshatch] (9, -1) rectangle ++(1, 0.5) node[midway,fill=white] { \textellipsis };

        \draw[thick,dashed] (1, 0) to[out=-90,in=0] (0.5, -0.25) -- (-0.5, -0.25) to[out=180,in=90] (-1, -0.5);
        \draw[thick,dashed] (4, 0) to[out=-90,in=0] (4.5,-0.25) -- (9.5,-0.25) to[out=0,in=90](10, -0.5);
    \end{tikzpicture}
    \caption{%
        Structure of an \prot{} packet.
        The lower row represents a zoomed-in view of \(\beta\), in which the hatched elements are encrypted to the current node and therefore not readable.
    }
    \label{fig:packet-structure}
\end{figure}

\subsection{Header creation}
\label{sec:format:header-creation}

We describe header creation as a function that takes as input the mix programs \(\mixprog_0\), \textellipsis, \(\mixprog_{\nu - 1}\) for each layer, as well as the public keys \(\pubkey_0\), \textellipsis, \(\pubkey_{\nu - 1}\) that are used to encrypt the onion layers.

The client samples a random \(x \unidraw \Zg^+\) and computes for each layer \(i\) the corresponding group element \(a_i\), the shared secret \(s_i\) and the blinding factor \(b_i\):
\begin{align*}
    \alpha_0 &= g^x, & s_0 &= y_0^x, & b_0 &= h_b(\alpha_0, s_0) \\
    \alpha_1 &= g^{x b_0}, & s_1 &= y_1^{x b_0}, & b_1 &= h_b(\alpha_1, s_1) \\
    \cdots & & \cdots & & \cdots & \\
    \alpha_{\nu - 1} &= g^{x b_0 b_1 \ldots b_{\nu - 2}}, & s_{\nu - 1} &= y_1^{x b_0 b_1 \ldots b_{\nu - 2}}, & b_{\nu - 1} &= h_b(\alpha_{\nu - 1}, s_{\nu - 1})
\end{align*}
Then, the client computes how the padding will look like at each node:
\begin{align*}
    \phi_0 = 0^0 \quad
    \phi_i = (\phi_{i - 1} \concat 0^{\abs{\mixprog_{i - 1}} + \secpar})
    \xor
    \rho(h_\rho(s_{i-1}))\substr{a_i}{b_i},
\end{align*}

where \(
a_i = j - \sum_{k=0}^{i-2} \left( \abs{\mixprog_k} + \secpar \right)
\) and \(
b_i = j + \abs{\mixprog_{i - 1}} + \secpar
\).

Finally, the client wraps the header from the inside out:
\begin{align*}
    \beta_{\nu - 1} &= (\mixprog_{\nu - 1} \xor \rho(h_\rho(s_{\nu - v1})))\substr{0}{\abs{\mixprog_{\nu - 1}}} \concat \phi_{\nu - 1} \\
    \gamma_i &= \mu(h_\mu(s_i), \beta_i) \\
    \hat{\beta}_i &=
        \mixprog_{i + 1} \concat
        \gamma_{i + 1} \concat
        (\beta_{i + 1})\substr{0}{j - \abs{\mixprog_{i + 1}}} \\
    \beta_i &= \hat{\beta} \xor \rho(h_{\rho}(s_i))\substr{0}{j} \\
    \gamma_i &= \mu(h_\mu(s_i), \beta_i)
\end{align*}
The resulting \(\eta = (\alpha_0, \beta_0, \gamma_0)\) is the final header, and the resulting \(s_i\) can be used to pre-process the payload.

\subsection{Header processing}
Once a mix node receives a packet, it first verifies its structure by checking that \(((\alpha, \beta, \gamma), \payload) \in (\group \times \byte^j \times \macSet) \times \byte^l\).
The node then computes the shared secret \(s = \alpha^{x}\) using its private key \(x\).
To ensure that the received packet is fresh, the node checks whether the replay tag \(h_\tau(s)\) appears in its table of tags already seen.
If the packet is not fresh, it is discarded.

Next, the node verifies the integrity of the header by computing the tag \(\gamma' = \mu(h_\mu(s), \beta)\).
If \(\gamma \ne \gamma'\), the packet is discarded.

Then, the node decrypts one layer of \(\beta\) as \(B' = \beta \oplus \rho(h_\rho(s))\) to obtain the processing logic for this hop via the mix program \(\mixprog\), and the message authentication tag \(\hat\gamma\) for the next hop.
The node does this by locating the \instr{Stop} instruction in \(B'\).
The bytes until there are \(\mixprog\), while the following \(\kappa\) bytes are \(\hat\gamma\).

Finally, the node makes the values from the header accessible to the mix program by loading them into predetermined registers.
Namely, it loads \(\alpha\), \(\beta\), \(\gamma\), \(\delta\), the processed versions of \(\alpha\), \(\beta\), and the next-hop MAC \(\hat\gamma\).

\subsection{Postprocessing}

For each outgoing packet, the node pads \(\beta\) and \(\delta\) to the correct length.
The padding is computed deterministically from the shared secret \(s\) and the index \(i\) of the outgoing packet as
\begin{align*}
    \beta' = \beta \concat \rho(h_\rho(s, s_\beta, i)) \quad\text{and}\quad
    \delta' = \delta \concat \rho(h_\rho(s, s_\delta, i)),
\end{align*}
where \(s_\beta, s_\delta \in \byte\) are arbitrarily chosen but fixed salt bytes.

\section{Mix programs}
\label{sec:programs}
In this section, we describe how \prot{} can be used to emulate existing mix formats at the example of Sphinx (\Cref{sec:programs:sphinx}) and PolySphinx (\Cref{sec:programs:polysphinx}).
We include the full mix programs in \Cref{apx:programs}.

After preprocessing, the mix node holds the the shared secret, current and next header elements, and the payload in their respective registers (see \Cref{sec:format}) and has the mix program loaded.

\subsection{Sphinx}
\label{sec:programs:sphinx}

As \prot{} closely builds upon Sphinx's package format, most of the required processing for Sphinx is already done in \prot{}'s pre- and postprocessing.
This includes the derivation of the shared secret, derivation of next headers, padding, and assembly of the outgoing packet.
As part of the mix program, an intermediate node first derives the decryption key by concatenating the shared secret with a salt, then \instr{Hash}ing the result.
After, \instr{Decrypt} is called, and finally postprocessing is initialized with a \instr{Forward} call.

If the node is an exit node, the node additionally has to extract the plaintext and receiver address from the decrypted payload.
At this point, the payload should have form \(\delta = 0^\kappa \parallel \Delta \parallel m\), where \(\Delta\) is the receiver's address and \(m\) is the message plaintext.
The node can extract the first \(\kappa\) bytes via \instr{CutBytes}, then create known zeroes with \instr{CreateZeroes}, and compare the two with \instr{IsEqual}.
Finally, the receiver address is extracted from the payload with a call of \instr{CutBytes} and message delivery is initialized with a \instr{Forward} call.

\subsection{PolySphinx}
\label{sec:programs:polysphinx}

There are three main differences between Sphinx and PolySphinx.
First, clients can include multiple next headers for nodes on the path.
If a node receives a packet with multiple headers, it uses them to send replicas of the packet onwards.
Second, to enable replication without linkability, payloads are encrypted along the path rather than decrypted.
Third, the exit node receives the key material needed to remove the encryption layers in form of a random seed that is used to generate the \enquote{key tree.}
The sender includes a distinct mix program for each type of node: simple forwarding, replication, and exit node.
\begin{itemize}
    \item The forwarding node's mix program is analogous to the Sphinx one, but replaces the \instr{Decrypt} call with an \instr{Encrypt} one.
    \item The replication node's mix program includes a \instr{For} loop with the fixed replication factor.
        Inside the loop, each next header is extracted via a series of \instr{CutBytes} calls, and then used to prepare and forward a replica of the packet.
    \item The exit node's mix program constructs the decryption keys based on the provided seed via a series of \instr{Hash} calls.
\end{itemize}

\section{Security argument}
\label{sec:security}
The following section contains our security argument for \prot{}.
Classical mix formats have provable privacy guarantees.
Through the formalized privacy properties \emph{Layer Unlinkability} (LU) and \emph{Tail Indistinguishability} (TI)~\cite{kuhn2020}, one can prove that malicious mix nodes cannot link senders to receivers.
A proof of LU and TI requires that the behavior of honest mix nodes is fixed, as is trivially the case in classical mix formats.
This is not the case in \prot{}, as the sender defines the behavior of the mix node through the instructions provided in each packet.
\prot{} thus cannot guarantee LU and TI for arbitrary mix programs.

Instead, we consider security in \prot{} in three aspects.
First, given a fixed mix program of just a \instr{Forward} instruction, we can prove that \prot{} achieves adapted versions of LU and TI.
In these adaptions, we only consider the packet header and omit the payload, as OmniSphinx treats it as a black box.
The actual proof then follows Sphinx's security proof~\cite{scherer2024a}:
The confidentiality of the shared secret follows from the Diffie-Hellman assumption, the confidentiality of the remaining header from the security of the pseudorandom generator \(\rho\), and the integrity follows from the MAC \(\mu\).
\ifextended
    We provide definitions and proofs of LU and TI below in \Cref{sec:security:base}. 
\else
    We refer to the extended version of this paper\footnote{\url{https://github.com/kit-ps/OmniSphinx/blob/esorics/Extended\%20Version.pdf}} for definitions and proofs of LU and TI.
\fi

Second, to verify that a given mix program is \enquote{secure,} we introduce \emph{information flow analysis} in the context of anonymous communication in \Cref{sec:security:infflow}.
Information flow analysis has been widely used since the 1970s to argue about the confidentiality guarantees of programs~\cite{DBLP:journals/cacm/DenningD77,sabelfeld2003language}.
As an example, we apply information flow analysis to \prot{}'s emulation of Sphinx.

Finally, as mix nodes execute programs given by users, we need to ensure that malicious users cannot harm mix nodes.
We discuss this aspect in \Cref{sec:security:node}.

\ifextended
\subsection{Base Format Unlinkability}
\label{sec:security:base}

\paragraph{Instruction Onion Routing Scheme}
We now adapt the classical onion-routing scheme to the \prot{} setting with mix programs.  
The resulting \emph{Instruction Onion Routing Scheme} captures the cryptographic behaviour of \prot{} while abstracting away from the semantics of mix programs.  
Instruction Onion Routing Scheme consists of the same high-level algorithms as classical Onion Routing Schemes, with the modifications that we introduce mix programs to the algorithms.

Formally, a Instruction Onion Routing Scheme consist of the following probabilistic polynomial-time algorithms:

\begin{itemize}
  \item $\mathsf{Gen}(1^\kappa, p, P) \to (PK, SK)$:  
  Generates a key pair for mix node $P \in N$, given the security parameter $\kappa$ and public parameters $p$.

  \item $\mathsf{FormOnionHeader}(i, \mathcal{R}, I,  \mathcal{P}, PK_\mathcal{P}) \to O_i$ with $\mathcal{P} = (P_1,\dots,P_n,\mathit{Receiver})$:  
  Used by senders to build onion headers. $i$ is the onion layer index, the case $i>1$ is relevant only in formal property definitions, not in practice.  
  $\mathcal{R}$ is the randomness for the algorithm, $\mathcal{P}$ is the path, $I$ is a list of mix programs to embed, and $PK_\mathcal{P} = (PK_1, \dots, PK_n)$ the public keys of the mix nodes in the path.

  \item $\mathsf{ProcOnionHeader}(SK, O, P) \to (O', P')$:  
  Used by mix nodes to process onion headers. $SK$ is the secret key of $P$ and $O$ is the onion header to process. In our case  $O'$ is the output onion header and $P'$ the next hop on the path. ProcOnionHeader is the preprocessing of OmniSphinx without interpreting the whole mix program. Instead of interpreting the mix program, we execute the forward instruction to the next hop, which is written in the mix program.
  In case of an error, $(O',P') = (\bot,\bot)$.  
  When processing the final onion: $\mathsf{ProcOnionHeader}(SK, O_n, P_n) \to (1,R)$.

  \item $\mathsf{ROnionHeader}(i, O, I,  \mathcal{R}, \mathcal{P}, PK_\mathcal{P}) \to b$:  
  Used in definitions to identify onion header.  
  If the header $O$ matches the header of the $i$-th onion built by $\mathsf{FormOnionHeader}$ given these parameters, then $b=1$. Otherwise $b=0$.  
\end{itemize}

\paragraph{Instruction Layer Unlinkability}
In the following, we introduce \emph{Instruction Layer Unlinkability (ILU)}, which adapts the notion of Layer Unlinkability (LU) to the setting of OmniSphinx. 
All modifications compared to \cite[Def. 4]{kuhn2020} are highlighted in \textcolor{blue}{\textit{blue italics}}.

\begin{definition}[Instruction Layer Unlinkability (ILU)]
Let {\color{blue}$(\mathsf{FormOnionHeader}, \\\mathsf{ProcOnionHeader},\mathsf{ROnionHeader})$}
be an {\color{blue}\emph{instruction}} onion routing scheme. We say that the {\color{blue}\emph{instruction}} onion routing scheme satisfies \emph{Instruction Layer Unlinkability (ILU)} if no probabilistic polynomial-time adversary $\mathcal{A}$ can win the following game with more than negligible advantage:

\begin{enumerate}
    \item \textbf{Setup:} The challenger generates key pairs for the honest participants:  
\[
(SK_s, PK_s) \leftarrow \mathsf{Gen}(1^\kappa) \quad \text{for the sender } P_s,
\]  
\[
(SK_H, PK_H) \leftarrow \mathsf{Gen}(1^\kappa) \quad \text{for the honest mix node } P_H,
\]  
The adversary $\mathcal{A}$ is given the public keys and the identities of the nodes $P_s, P_H$.  $\mathcal{A}$ may generate keys for all other nodes of its choice.  

    \item \textbf{Oracle access:} $\mathcal{A}$ may submit any number of onion headers $O$ of its choice for a \textsf{Proc} request of $P_H$ or $P_s$ to the challenger. When asked for \textsf{Proc}$(P_H, O)$  
    the challenger checks whetever the header $\eta$ is on the $\eta_H$-list. If not, it responds with  {\color{blue}{$\textsf{ProcOnionHeader} (SK_H, O, P_H)$}} and stores $\eta$ on the $\eta_H$-list.  (Similar for requests to $P_s$)
    
    \item \textbf{Challenge:} $\mathcal{A}$ submits:
    \begin{itemize}
        \item a path $P = (P_1, \dots, P_j, \dots, P_{n}, R)$ with $P_j = P_H$ and a receiver $R$.
        \item public key $(PK_i)$ for all nodes  $P_i$ on the path with $i \ne j$,
        \item \color{blue}{\textit{a set of mix programs $I = (I_1, \dots, I_{n})$ to be embedded in the packet.}}
    \end{itemize}

    \item \textbf{Verification:} The challenger checks that the chosen paths are acyclic, the mix node identifiers are valid and that the same key is chosen if the mix node identifiers are equal, then samples a random bit $b \in \{0,1\}$.

    \item \textbf{Onion Construction:}  
    The challenger creates:
    \begin{itemize}
        \item an onion header $O_1\leftarrow \color{blue}\textsf{FormOnionHeader}(1, \mathcal{R}, I, \mathcal{P}, \mathsf{PK}_\mathcal{P})$ with the adversary’s path and instructions,  
        \item a random onion header $\bar{O}_1 \leftarrow \color{blue}\textsf{FormOnionHeader}(1, \mathcal{R}, \bar{I}, \bar{\mathcal{P}}, \mathsf{PK}_{\bar{\mathcal{P}}})$, with the first part of the path and a random receiver $\bar{R}$: $\bar{\mathcal{P}} = (P_1, \dots , P_j, \bar{R})$, \color{blue}{\textit{and the first part of the mix programs $\bar{I} = \{ I_1, \dots, I_{j+1}\}$}}.
    \end{itemize}

    \item \textbf{Challenge Output:}  
    \begin{itemize}
        \item If $b=0$: the challenger gives $O_1$ to $\mathcal{A}$.
        \item If $b=1$: the challenger gives $\bar{O}_1$ to $\mathcal{A}$.
    \end{itemize}

    \item \textbf{Post-Challenge Oracle Access:}  
    If $b=0$, the challenger answers all oracle queries exactly as in step~2.  

    If $b=1$, the challenger answers all oracle queries as in step~2, 
    \emph{except} in the following cases, where the challenger simulates fresh onions instead of processing the challenge onion:  

\begin{itemize}
    \item If $j < n$:  
    For a query $\mathsf{Proc}(P_H, O = \eta)$ such that $\eta$ not in $\eta_H$ , 
    \[
    {\color{blue}\mathsf{ROnionHeader}(j, O, \bar{\mathcal{R}}, \bar{I}, \bar{\mathcal{P}}, \mathsf{PK}_{\bar{\mathcal{P}}})}
    = \textsf{True}
    \quad \text{and} 
    \]
    \[
    {\color{blue}\mathsf{ProcOnionHeader}(SK_H, O, P_H)} \neq (\bot,\bot),
    \]
    the challenger outputs $(P_{j+1}, O_c)$ where  
    \[
        O_c \leftarrow \mathsf{FormOnionHeader}(1, \mathcal{R}', I',  \mathcal{P}', \mathsf{PK}_{\mathcal{P}'})
    \]  
    with honestly chosen randomness $R'$, {\color{blue}instructions $I' = (I_{j+1}, \dots, I_n)$}, and path $P' = (P_{j+1}, \dots, P_n, R)$ and adds $\eta$ to the $\eta_H$-list.  

    \item If $j = n$:  
    For a query $\mathsf{Proc}(P_H, O =\eta)$ such that $\eta$ not in $\eta_H$ , 
    \[
    {\color{blue}\mathsf{ROnionHeader}(j, O, \bar{\mathcal{R}}, \bar{I}, \bar{\mathcal{P}}, \mathsf{PK}_{\bar{\mathcal{P}}})} = \textsf{True} 
        \quad \text{and} 
    \]
    \[
        {\color{blue}\mathsf{ProcOnionHeader}(SK_H, O, P_H)} \neq (\bot,\bot),
    \]  
    the challenger outputs $(1,\bot)$ and adds $\eta$ to the $\eta_H$-list.  
\end{itemize}

    \item \textbf{Guess:} $\mathcal{A}$ outputs a guess $b'$.
\end{enumerate}

ILU is achieved if for all PPT adversaries $\mathcal{A}$:
\[ \left| \Pr[b = b'] - \frac{1}{2} \right| \leq \textsf{negl}(\kappa) \]
\end{definition} 

To demonstrate that OmniSphinx satisfies Instruction Layer Unlinkability, we follow the common approach of constructing a sequence of hybrid games. 
Starting from the real ILU game with $b = 0$, we gradually transform the adversary’s view until it becomes indistinguishable from the $b=1$ scenario, where all information about the input onion has been replaced by random values. 
Each hybrid step is justified by the cryptographic properties of the primitives used in OmniSphinx.
Together, these arguments establish that the adversary cannot achieve a non-negligible advantage, thus proving ILU for the meta-construction.
This proof closely follows and adapts Kuhn et al.'s Layer Unlinkability proof of Sphinx~\cite{kuhn2020}.

\begin{theorem}
OmniSphinx achieves Instruction Layer Unlinkability against a PPT adversary under the GDH assumption.
\end{theorem}

\begin{proof}

We construct a sequence of hybrids $H_0,\dots,H_{11}$ to show that the real game ($H_0$) is computationally indistinguishable from the ideal game ($H_{11}$). 

\begin{description}
  \item[Hybrid $H_0$:] The real ILU game with $b=0$.

    \item[Hybrid $H_1$:] 
    In this hybrid, all keys derived from $s_j$ at the honest node $P_j$ (e.g., the PRG key $h_\rho$ for $\beta_j$, the MAC key $h_\mu$ for $\gamma_j$ and the blinding factor $h_b$ for $\alpha_j$) are replaced with uniform random strings. 
    
    \emph{Argument} $H_0 \approx H_1$: $\alpha$ in OmniSphinx is constructed in the same way as in Sphinx. Scherer et al. formally proved that this step is indistinguishable under the Gap Diffie–Hellman (GDH) assumption in the random oracle model \cite{scherer2024a}.

    \item[Hybrid $H_2$:] 
    The \textsf{Proc} oracle returns $\bot$ on every request with $\alpha_{j-1}$ in its header, except if the rest of the header also matches the expected header of the challenge onion. This ensures that only the challenge onion is recognized for challenge processing by the honest node. 
    
    \emph{Argument} $H_1 \approx H_2$: If an adversary can distinguish between $H_1$ and $H_2$, we can construct an adversary $\mathcal{A'}$ that breaks the sEUF-CMA security of the MAC $\mu$, which is modeled as a secure PRF. Any valid \textsf{Proc} request with $\alpha_{j-1}$ in its header must include a valid tag $\gamma_{j-1}$. To notice a difference between the two hybrids, a distinguisher must submit such a request with a modified $\gamma$. Such a request directly constitutes a MAC forgery, contradicting the assumed security of $\mu$.

    \item[Hybrid $H_3$:]
    In the honest node challenge processing, $H_3$ always produces the same challenge header (the one belonging to the challenge onion's layer $O_j$) without actually processing the header input the relay is given.
    
    \emph{Argument} $H_2 \approx H_3$: In $H_1$, the honest node only performs the challenge processing steps on headers that matche the challenge header exactly. Thus, both hybrids always output the identical challenge header for the challenge onion.

    \item[Hybrid $H_4$:] 
    Replace the PRG output $\rho(h_\rho(s_{j-1}))$ used to mask $\beta_j$ with uniform random bits. 
    
    \emph{Argument} $H_3 \approx H_4$: If a distinguisher $\mathcal{D}$ can tell apart $H_2$ and $H_3$, we can construct an adversary $\mathcal{A}$ that breaks the pseudorandomness of the PRG. Since the output $\rho$ is only used as a one-time mask for $\beta_j$ and oracle queries with the challenge header are restricted (from $H_2$), distinguishing $\rho$ from uniform directly contradicts the pseudorandomness of the PRG.

    \item[Hybrid $H_5$:] 
    In $\beta_{j-1}$, replace all non-padding fields (the mix program $\textsf{mp}_{j-1}$, the tag $\gamma_j$, etc.) with uniform random strings, while keeping the filler $\phi_j$ unchanged. The replacement is a mix program that has a forward to the receiver $R$, making $P_j$ the last node on the path.
    
    \emph{Argument} $H_4 \approx H_5$: 
    Since the mask $\rho(h_\rho(s_{j-1}))$ is uniform and used only once, the encryption of these fields is equivalent to a one-time pad. Thus, a distinguisher between $H_3$ and $H_4$ could be transformed into an adversary that breaks the IND-CPA security of the OTP.
 
    \item[Hybrid $H_6$:] 
    In $H_4$, the challenge onion contains the nested padding $\Phi_0, \dots, \Phi_{j-1}$ inside $\beta_j$'s filler $\Phi_j$. In $H_5$, we replace $\Phi_j$ with a uniform random string of the same length. 
    
    \emph{Argument} $H_5 \approx H_6$: 
    In $H_4$, the filler string $\Phi_j$ is computed as 
    $\Phi_j = \rho(h_\rho(s_{j-1})) \oplus \{\Phi_{j-1} \| 0_{|I_j|}\}$. 
    Since $\rho(h_\rho(s_{j-1}))$ is uniform (by $H_3$), $\Phi_j$ is indistinguishable from a random string. 
    Thus, a distinguisher between $H_4$ and $H_5$ could be used to build an adversary against the IND-CPA security of the one-time pad.

    \item[Hybrid $H_7$:] 
    Resample the DH element $\alpha_j$: choose a fresh $x' \leftarrow \mathbb{Z}_q^*$ at random, set $\alpha_j := g^{x'}$ and $s_j := y_j^{x'}$. 
    
    \emph{Argument} $H_6 \approx H_7$: In $H_1$ we already randomized $b_{j-1}$ into a uniformly distributed element of $\mathbb{Z}_q^*$. Before $H_6$ we have $\alpha_j = \alpha_{j-1}^{b_{j-1}}$. Since $b_{j-1}$ is uniform, $\alpha_{j-1}^{b_{j-1}}$ is identically distributed to $g^{x'}$ for fresh $x' \leftarrow \mathbb{Z}_q^*$. The same reasoning applies to all later $\alpha$-values and secrets.

    \item[Hybrid $H_8$:] 
    Construct fresh prefix layers $O'_0,\dots,O'_{j-1}$ following the same path and instructions as the original $O_0,\dots,O_{j-1}$, and attach them before $O_j$, such that $\beta'_{j-1}$ is formed with $\beta_j$ in its contents. Randomize all random oracle outputs. Set $\alpha_j = (\alpha'_{j-1})^{b'_{j-1}}$, with the corresponding secrets derived analogously.
    
    \emph{Argument} $H_7 \approx H_8$: The newly generated layers $O'_0,\dots,O'_{j-1}$ are never revealed to the adversary, hence their construction is invisible except for the way $\alpha_j$ is computed. However, as argued in $H_5 \approx H_6$, $\alpha_j$ has the same distribution in both hybrids.

    \item[Hybrid $H_9$:] Replace $\rho(h_\rho(s_{j-1}'))$ with a random string. 
  
    \emph{Argument} $H_8 \approx H_9$: See $H_2 \approx H_3$.

    \item[Hybrid $H_{10}$:] 
    Replace the filler $\phi_j$ with the canonical PRG-derived padding, so that $\Phi_j$ is formed as the $j$-th layer of padding in $O'_0,\dots,O_j,\dots,O_{n-1}$, i.e., 
    $\Phi_j = \rho(h_\rho(s'_{j-1})) \oplus \{\Phi_{j-1} \| 0_{|I_j|}\}$. 
    
    \emph{Argument} $H_9 \approx H_{10}$: In $H_5$, $\Phi_j$ was replaced with a uniform random string. So this argument is analogous to $H_4 \approx H_5$.

  \item[Hybrid $H_{11}$:] 
  Replace the randomized oracle outputs for $s'_{j-1}$ with the actual outputs $h_*(s'_{j-1})$. 
  
  \emph{Argument} $H_{10} \approx H_{11}$: See $H_0 \approx H_1$.

  \item[Hybrid $H_{12}$:] 
  Rewind all temporary modifications made in the previous hybrids in the reverse order: $H_8, H_6, H_2, H_1$.
  
  \emph{Argument} $H_{11} \approx H_{12}$: Apply the previous arguments in reverse.
\end{description}

By transitivity of indistinguishability, $H_0 \approx H_{12}$, proving that \prot{} satisfies ILU.

\end{proof}

\paragraph{Instruction Tail Indistinguishability}
In the following, we introduce \emph{Instruction Tail Indistinguishability (ITI)}, which adapts the notion of Tail Indistinguishability (TI) to the setting of OmniSphinx. 
All modifications from \cite[Def. 3]{kuhn2020} are highlighted in \textcolor{blue}{\textit{blue italics}}.

\begin{definition}[Instruction Tail Indistinguishability (ITI)]
Let {\color{blue}$(\mathsf{FormOnionHeader}$, \\ $\mathsf{ProcOnionHeader}, $$\mathsf{ROnionHeader})$} be an {\color{blue}\emph{instruction}} onion routing scheme. We say that the {\color{blue}\emph{instruction}} onion routing scheme satisfies \emph{Instruction Tail Indistinguishability (ITI)} if no probabilistic polynomial-time adversary $\mathcal{A}$ can win the following game with more than negligible advantage:

\begin{enumerate}
\item \textbf{Setup:} The challenger generates key pairs for the honest participants:  
\[
(SK_s, PK_s) \leftarrow \mathsf{Gen}(1^\kappa, \mathcal{P}, P_s) \quad \text{for the sender } P_s,
\]  
\[
(SK_H, PK_H) \leftarrow \mathsf{Gen}(1^\kappa, \mathcal{P}, P_H) \quad \text{for the honest mix node } P_H,
\]  
The adversary $\mathcal{A}$ is given the public keys and the identities of the nodes $P_s, P_H$.  $\mathcal{A}$ may generate keys for all other nodes of its choice.  

    \item \textbf{Oracle access:} $\mathcal{A}$ may submit any number of onion headers $O$ of its choice for a \textsf{Proc} request of $P_H$ or $P_s$ to the challenger. When asked for \textsf{Proc}$(P_H, O)$  
    the challenger checks if the header $\eta$ is on the $\eta_H$-list. If not, it responds with  {\color{blue}{$\textsf{ProcOnionHeader} (SK_H, O, P_H)$}} and stores $\eta$ on the $\eta_H$-list.  (Similar for requests to $P_s$)

\item \textbf{Challenge:} $\mathcal{A}$ submits:
\begin{itemize}
    \item a path $\mathcal{P} = (P_1, \dots, P_j, \dots, P_{n}, R)$ with $P_j = P_H$ honest and a Receiver $R$,
    \item public keys $(PK_i)$ for all nodes $P_i$ on the path with $i \ne j$,
    \item \color{blue}{\textit{a set of mix programs $I = (I_1, \dots, I_{n})$ to be embedded in the packet.}}
\end{itemize}

\item \textbf{Verification:} The challenger checks that the chosen path is acyclic, the mix node identifiers are valid and that the same key is chosen if the mix node identifiers are equal, then samples a random bit $b \in \{0,1\}$.

\item \textbf{Onion Construction:}  
The challenger creates:
\begin{itemize}
    \item an onion header $O_{j+1} \leftarrow \color{blue}\textsf{FormOnionHeader}(j+1, \mathcal{R}, I, \mathcal{P}, \mathsf{PK}_\mathcal{P})$ with the adversary’s path, instructions and honestly chosen randomness $\mathcal{R}$,
    \item a random onion header $\bar{O}_1 \leftarrow \color{blue}\textsf{FormOnionHeader}(1, \bar{\mathcal{R}}, \bar{I}, \bar{\mathcal{P}}, \mathsf{PK}_{\bar{\mathcal{P}}})$, with the path from the honest mix node $P_H$ to the corrupted final mix node $\bar{\mathcal{P}} = (P_{j+1}, \dots, P_n, R)$, {\color{blue}$\bar{I} = (I_{j+1}, \dots, I_{n+1})$}, and honestly chosen randomness $\bar{\mathcal{R}}$.
\end{itemize}

\item \textbf{Challenge Output:}  
\begin{itemize}
    \item If $b=0$: the challenger gives $O_{j+1}$ to $\mathcal{A}$.
    \item If $b=1$: the challenger gives $\bar{O}_1$ to $\mathcal{A}$.
\end{itemize}

\item \textbf{Post-Challenge Oracle Access:}  
The challenger continues to answer oracle queries exactly as in step~2, independently of the bit $b$.

\item \textbf{Guess:} $\mathcal{A}$ outputs a guess $b'$.
\end{enumerate}

ITI is achieved if for all PPT adversaries $\mathcal{A}$:
\[
\left| \Pr[b = b'] - \tfrac{1}{2} \right| \leq \mathsf{negl}(\kappa).
\]
\end{definition}

Just as in Instruction Layer Unlinkability, we demonstrate that OmniSphinx satisfies Instruction Tail Indistinguishability by following the common approach of constructing a sequence of hybrid games. Starting from the real ITI game with $b = 0$, we gradually transform the adversary’s view until it becomes indistinguishable from the $b=1$ scenario. We gradually transform the $O_j$ onion header into the $\bar{O}_0$ onion header in successive hybrids. This proof closely follows and adapts Kuhn et al.'s Tail Indistinguishability proof of Sphinx~\cite{kuhn2020}.

\begin{theorem}
\prot{} achieves Instruction Layer Unlinkability against any PPT adversary under the GDH assumption.
\end{theorem}

\begin{proof}
We define hybrids $H_0,\dots,H_5$ and show that the real experiment $H_0$ is
computationally indistinguishable from the truncated-onion experiment $H_5$.

\begin{description}

\item[Hybrid $H_0$:] 
The real ITI game with challenge bit $b=0$.

\item[Hybrid $H_1$:] 
If $j=0$, no prefix exists and we are already in the truncated case.  
Otherwise, when constructing $O_j$, we randomise the Diffie-Hellman component
by sampling a fresh exponent $x' \in \mathbb{Z}_q^*$ and setting
\[
   \alpha_j := g^{x'}, \qquad s_j := y_j^{x'}.
\]

\emph{Argument $H_0 \approx H_1$:}  
Follows from the DH/KEM security argument used in Hybrid~$H_7$ of the original ILU proof.

\item[Hybrid $H_2$:] 
When generating the \prot{} header, we replace $h_\rho(s_{j-1})$ with a
uniform $\kappa$-bit string.

\emph{Argument $H_1 \approx H_2$:}  
Since $s_{j-1}$ is indistinguishable from a random group element, its hash in the random oracle is uniform.

\item[Hybrid $H_3$:] 
Replace $\rho(h_\rho(s_{j-1}))$ with a uniformly random string of equal length.  
Consequently, the value $\Phi_j$, which is XORed against this output, becomes uniform as well.

\emph{Argument $H_2 \approx H_3$:}  
Analogous to Hybrid~$H_4$ of the ILU proof; PRG outputs can be replaced by uniform strings.

\item[Hybrid $H_4$:] 
When constructing $\beta_{n-1}$, we extend its random portion by
\[
   \sum_{k=0}^{j} |\textsf{mp}_k| + \kappa
\]
fresh random bits and truncate $\Phi_{n-1}$ by the same amount.

\emph{Argument $H_3 \approx H_4$:}  
This redistribution only moves already-uniform randomness between header fields
and therefore leaves the distribution unchanged.

\item[Hybrid $H_5$:] 
We no longer generate any of the prefix layers
\[
   \alpha_0,\dots,\alpha_{j-1},\ 
   \beta_0,\dots,\beta_{j-1},\ 
   \gamma_0,\dots,\gamma_{j-1},\ 
   \Phi_0,\dots,\Phi_{j-1}.
\]

\emph{Argument $H_4 \approx H_5$:}  
These values are unused in $H_4$ and never appear in the adversary's view, so omitting them is undetectable.

\end{description}

In $H_5$, the packet is distributed exactly as in the $b=1$ case of ITI: all
dependencies on the prefix have been eliminated.  
Since each transition is indistinguishable, we conclude $H_0 \approx H_5$.  
Thus \prot{} satisfies ITI.
\end{proof}

\fi

\subsection{Information Flow Analysis}
\label{sec:security:infflow}
Information flow analysis tracks the observability of data during protocol execution.
Program values are classified either \good{} or \bad{} prior to execution.
Execution may change the state of a value.
The goal is to ensure that no \bad{} information is visible in the resulting packet.

For analysis, we model the mix program as a directed dependency graph, where the nodes represent values occurring during program execution.
Values are input parameters, registers initialized during preprocessing, intermediate results, or final outputs.
Each node has a corresponding security label \(\ell \in \left\{\good{}, \bad{}\right\}\).
Edges in the graph are determined by the mix program's instructions:
An instruction producing \(y \gets f(x_1,\dots,x_k)\) induces an edge from each \(x_{i\in[1,k]}\) to \(y\).
The label \(\ell(y)\) is determined by the labels \(\ell(x_i)\) for \(i\in[1,k]\) as well as the concrete instruction \(f(\cdot)\):

For structural and arithmetic instructions (\instr{XOR}, \instr{Concat}, \instr{Copy}, \textellipsis) the output retains information about the input.
Thus, if any input is \bad{}, the output is also \bad{}:
\[
    \ell(y) := \begin{cases}
        \bad{} & \text{if } \exists x_{i\in[1,k]}: \ell(x_i) = \bad{}\\
        \good{} & \text{else}
    \end{cases}
\]
The cryptographic operations \instr{Encrypt}, \instr{Decrypt}, and \instr{MAC} take as input a key and a message/ciphertext.
If the key is \bad{}, the operation ensures that the output is computationally indistinguishable from randomness and thus can be labeled \good{}.
If the key is \good{}, the operation can be re-executed by the adversary, and the output inherits the message/ciphertext's label:
\begin{align*}
    \text{For } f &\in \{\instr{Encrypt}, \instr{Decrypt}, \instr{MAC}\} \text{ with } \textit{out} \gets f(\textit{key}, \textit{in}):\\
    \ell(\textit{out}) &:= \begin{cases}
        \good{}, & \text{if } \ell(\textit{key}) = \bad{}\\
        \ell(\textit{in}) & \text{else}
    \end{cases} 
\end{align*}
Similarly, a \instr{PRF}'s output is \good{} given that its seed is \bad{}, as it computationally indistinguishable from randomness.

Within a \prot{} program, the inputs of a \instr{ForLoop} are always pre-determined by the sender.
Any program can thus be serialized for analysis.
Given a \(\instr{ForLoop}(d,i)\) instruction, serialization repeats the following \(d\) instructions \(i\) times and removes the original \instr{ForLoop} instruction.

Finally, the mix program contains no unintended information flow, if each input of any \instr{Forward} instruction the program contains is \good{}.

\paragraph{Analysis of Sphinx}
As an example of how to apply information flow analysis to mix programs, we analyze the Sphinx program for \prot{}.

The Spinx mix program is secure, if no \bad{} information is forwarded by the mix node.
If \bad{} information is forwarded, incoming packets can be linked to outgoing ones and the mix network's privacy guarantees no longer hold.
As the first step in the analysis, we need to label the inputs.
The incoming packet (\(\alpha, \beta, \gamma, \delta\)) is \bad{}, as it was computed by and is therefore known to the previous node.
The shared secret \(s\) is also labeled \bad{}, as disclosing it would enable an adversary to directly derive the outgoing packet from the incoming one themselves.
The next headers \(\alpha', \beta', \gamma'\) are labeled \good{}, as they result from cryptographic operations and are not distinguishable from randomness, assuming \(s\) is know known.
The \textit{salt} is labeled \good{}, as it is a publicly known constant.
The \text{next} node is labeled \good{}, as it is inherently revealed by the node when forwarding the packet.

The shared secret \(s\) and \textit{salt} are concatenated with \instr{ConcatBytes} and \instr{Hash}ed.
Neither operation changes the \bad{} label of \(s\), so the output of \instr{Hash} remains \bad{}.
The \bad{} \(\delta\) and the \bad{} hash output are input into \instr{Decrypt}.
Due to the cipher's confidentiality guarantees, the output of \instr{Decrypt} changes to \good{}.
The \textsf{next} node is input into \instr{Load}, which does not change the label.
Finally, \instr{Forward} receives the output of \instr{Decrypt}, \instr{Load}, as well as \(\alpha'\), \(\beta'\), and \(\gamma'\).
All inputs of \instr{Forward} are labeled \good{}, therefore the mix program contains no unintended information flow.
A graph representation of this information flow analysis can be found in \Cref{fig:infoflow}.

\begin{figure}
    \centering
    \begin{tikzpicture}[
        innode/.style={draw,minimum width=1cm,minimum height=0.6cm},
        callnode/.style={draw,rounded corners,fill=gray!20,inner sep=2mm},
        labelnode/.style={inner sep=-0.3mm,fill=white},
        label2node/.style={inner sep=-0.4mm,rounded corners,fill=gray!20}
    ]
        \node[innode] (0) at (0,0) {\(\alpha\)};
        \node[innode] (1) at (1.2,0) {\(\beta\)};
        \node[innode] (2) at (2.4,0) {\(\gamma\)};
        \node[innode] (3) at (3.6,0) {\(\delta\)};
        
        \node[innode] (4) at (4.8,0) {\(\alpha'\)};
        \node[innode] (5) at (6,0) {\(\beta'\)};
        \node[innode] (6) at (7.2,0) {\(\gamma'\)};
        \node[innode] (7) at (8.4,0) {\(s\)};
        
        \node[innode] (8) at (9.6,0) {\textit{salt}};
        \node[innode] (9) at (10.8,0) {\textit{next}};

        \node[] (10) at (1.8,0.6) {incoming packet};
        \node[] (11) at (6.6,0.6) {preprocessing};
        \node[] (12) at (10.2,0.6) {mix program};

        \draw[] ([yshift=1mm]0.north west) |- (10) -| ([yshift=1mm]3.north east);
        \draw[] ([yshift=1mm]4.north west) |- (11) -| ([yshift=1mm]7.north east);
        \draw[] ([yshift=1mm]8.north west) |- (12) -| ([yshift=1mm]9.north east);

        \node[callnode] (13) at (9,-1) {\instr{ConcatBytes}};
        \draw[->] (7) -- ([xshift=-3mm]13.north);
        \draw[->] (8) -- ([xshift=3mm]13.north);
        
        \node[callnode] (14) at (9,-2) {\instr{Hash}};
        \draw[->] (13) -- (14);
        
        \node[callnode] (15) at (10.8,-3) {\instr{Load}};
        \draw[->] (9) -- (15);
        
        \node[callnode] (16) at (3.6,-2) {\instr{Decrypt}};
        \draw[->] (3) -- (16);
        \draw[->] (14) -- (16);
        
        \node[callnode] (17) at (6,-3) {\instr{Forward}};
        \draw[->] (4) -- ([xshift=-3mm]17.north);
        \draw[->] (5) -- (17.north);
        \draw[->] (6) -- ([xshift=3mm]17.north);
        \draw[->] (16) |- (17);
        \draw[->] (15) -- (17);

        %labels
        green!66!blue
        \node[labelnode] (l0) at ([xshift=-1mm,yshift=-1mm]0.north east) {\textcolor{red!66!yellow}{\faEye}};
        \node[labelnode] (l1) at ([xshift=-1mm,yshift=-1mm]1.north east) {\textcolor{red!66!yellow}{\faEye}};
        \node[labelnode] (l2) at ([xshift=-1mm,yshift=-1mm]2.north east) {\textcolor{red!66!yellow}{\faEye}};
        \node[labelnode] (l3) at ([xshift=-1mm,yshift=-1mm]3.north east) {\textcolor{red!66!yellow}{\faEye}};
        
        \node[labelnode] (l4) at ([xshift=-1mm,yshift=-1mm]4.north east) {\textcolor{green!66!blue}{\faEyeSlash}};
        \node[labelnode] (l5) at ([xshift=-1mm,yshift=-1mm]5.north east) {\textcolor{green!66!blue}{\faEyeSlash}};
        \node[labelnode] (l6) at ([xshift=-1mm,yshift=-1mm]6.north east) {\textcolor{green!66!blue}{\faEyeSlash}};
        \node[labelnode] (l7) at ([xshift=-1mm,yshift=-1mm]7.north east) {\textcolor{red!66!yellow}{\faEye}};
        
        \node[labelnode] (l8) at ([xshift=-1mm,yshift=-1mm]8.north east) {\textcolor{green!66!blue}{\faEyeSlash}};
        \node[labelnode] (l9) at ([xshift=-1mm,yshift=-1mm]9.north east) {\textcolor{green!66!blue}{\faEyeSlash}};
        
        \node[label2node] (l13) at ([xshift=-1mm,yshift=-1mm]13.north east) {\textcolor{red!66!yellow}{\faEye}};
        \node[label2node] (l14) at ([xshift=-1mm,yshift=-1mm]14.north east) {\textcolor{red!66!yellow}{\faEye}};
        \node[label2node] (l15) at ([xshift=-1mm,yshift=-1mm]15.north east) {\textcolor{green!66!blue}{\faEyeSlash}};
        \node[label2node] (l16) at ([xshift=-1mm,yshift=-1mm]16.north east) {\textcolor{green!66!blue}{\faEyeSlash}};
        
        \node[label2node] (ldone) at (17.south) {\textcolor{green!66!blue}{\faCheckCircle}};
    \end{tikzpicture}
    \caption{Information flow analysis for the Sphinx mix program in \prot{}. \textcolor{red!66!yellow}{\faEye} denotes a \bad{} label, while \textcolor{green!66!blue}{\faEyeSlash} denotes a \good{} label.}
    \label{fig:infoflow}
\end{figure}
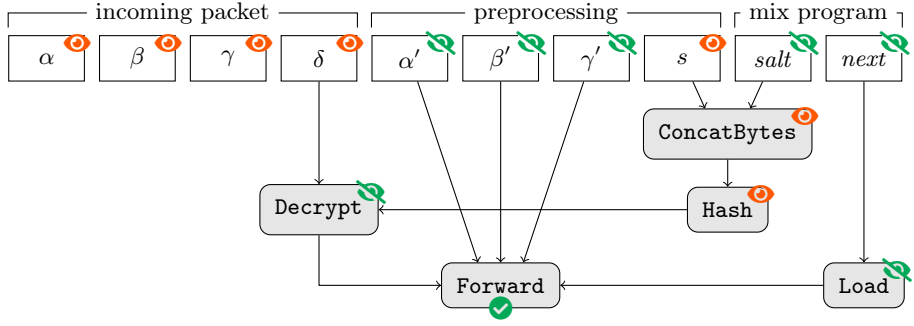

\subsection{Mix node security}
\label{sec:security:node}
In terms of mix node security, we identify three potential harms:
Malicious users can try to \emph{exfiltrate} secrets from the mix node, such as its private key or information about other messages.
Users can also try to \emph{control} the node in order to use it for other malicious purposes, like becoming a part of a botnet.
Finally, users can try to \emph{jam} the node to prevent it from further work.

\prot{} protects against exfiltration, as the mix programs have no access to secret key material:
The only value that depends on the secret key is the shared secret, which is different from packet to packet.
However, the (malicious) sender already knows the shared secret, thus this provides no benefit to them.
Thanks to the Diffie-Hellman key exchange, it is not possible to learn the secret key based on the shared secret.
All other values that a mix program has access to come directly from the packet, and reveal no information about other packets that the mix node processes.

\prot{} also protects against malicious control and jamming (denial-of-service):
The instruction set is deliberately restricted to prevent arbitrary interactions with other systems.
It does not allow custom network access or resource acquisition.
Additionally, infinite loops are not possible, and the execution time of each program is bounded (see also \Cref{sec:eval:computation} and \Cref{tab:computation:instruction}).
The node can further put limits on the execution time or memory consumption per program to prevent a client from using up the node's resources.

\section{Evaluation}
\label{sec:evaluation}
In the following, we provide an empirical evaluation of \prot{}.
We first describe our implementation of \prot{}, then evaluate performance in terms of bandwidth (\Cref{sec:eval:bandwidth}) and computation (\Cref{sec:eval:computation}).

\paragraph{Implementation}
We implemented \prot{} in Java.
For all cryptographic operations, we rely on BouncyCastle\footnote{\url{https://www.bouncycastle.org/}} v1.70.
We implement \instr{Hash} as SHA-256, \instr{Mac} as HMAC-SHA256, \instr{Encrypt} and \instr{Decrypt} as LIONESS~\cite{anderson1996}, and \instr{Prg} as an AES-CTR keystream.
As the group \(\group\) we use the \texttt{secp224r1} NIST elliptic curve.
We publish our implementation at \url{https://github.com/kit-ps/OmniSphinx}.

\subsection{Bandwidth}
\label{sec:eval:bandwidth}
Our goal is to determine how much bandwidth overhead emulation in \prot{} introduces versus the native mix formats.
\prot{} packets differ from other Sphinx-like packets in the size of their \(\beta\) header, where \prot{} includes the mix programs.
Within each instance, the size of \(\beta\) is determined by a public protocol parameter, otherwise packets can be linked via their header size.
We want to determine the minimum size of \(\beta\) in \prot{} to emulate state of the art mix formats and compare the resulting header size to the header size of the native format.

We compare \prot{} against Sphinx, AE-Sphinx, EROR, MultiSphinx, and PolySphinx.
Native header sizes are derived from the respective format's description.
For the minimum \prot{} \(\beta\), we implement a reference mix program for each tested format.
For MultiSphinx and PolySphinx, we assume a replication factor of \(p=3\).
For all experiments, we assume path lengths of 5 and a security parameter \(\kappa = \qty{16}{\byte}\).
We present the total header size.

For \prot{}, the header has to contain a mix program, whereas it only has to hold the routing information for the native protocols.
We thus expect the \prot{} header to be larger than the native headers in all cases.
We further expect the smallest difference in header size between \prot{} and Sphinx, as most of Sphinx's processing is already handled by \prot{}'s preprocessing, resulting in a short mix program (see \Cref{sec:programs:sphinx}).

\begin{table}[h]
    \centering
    \begin{tabularx}{0.8\textwidth}{lXXll}\toprule
         Mix format & Native & Emulated & \multicolumn{2}{l}{Difference}\\\midrule
         Sphinx & \pz{}\qty{205}{\byte} & \pz{}\qty{273}{\byte} & +\pz{}\qty{68}{\byte} & \textcolor{gray}{+\pz{}33\%} \\
         AE-Sphinx & \pz{}\qty{205}{\byte} & \pz{}\qty{465}{\byte} & +\qty{260}{\byte} & \textcolor{gray}{+127\%} \\
         EROR & \pz{}\qty{485}{\byte} & \pz{}\qty{725}{\byte} & +\qty{240}{\byte} & \textcolor{gray}{+\pz{}49\%} \\\addlinespace
         MultiSphinx (\(p=3\)) & \pz{}\qty{205}{\byte} & \pz{}\qty{489}{\byte} & +\qty{284}{\byte} & \textcolor{gray}{+139\%} \\ 
         PolySphinx (\(p=3\)) & \qty{1256}{\byte} & \qty{1545}{\byte} & +\qty{289}{\byte} & \textcolor{gray}{+\pz{}23\%} \\\bottomrule
    \end{tabularx}
    \caption{Minimum header size of \prot{} emulations versus existing mix formats.}
    \label{tab:bandwidth}
\end{table}

\Cref{tab:bandwidth} presents our results and confirms our expectations:
\prot{} headers are between \qty{68}{\byte} and \qty{289}{\byte} larger than native headers.
The smallest difference is observed with Sphinx.
Native MultiSphinx, AE-Sphinx and EROR re-use the MAC in the header to secure the integrity of the payload, however in OmniSphinx, this MAC is computed as part of the preprocessing and cannot be influenced.
Therefore, those formats have to include an extra MAC in the header, leading to an increased size.
Emulated PolySphinx headers have a large absolute increase over native ones, but since the headers are large to begin with, the relative increase is small.

While \prot{} produces significantly larger headers compared to native formats, the increase is small in absolute numbers.
To emulate all evaluated formats, a \(\beta\) size of \qty{1545}{\byte} has to be used for all packets.
Then, emulating Sphinx introduces the largest overhead with \qty{1340}{\byte} per packet over the native format assuming five hops. 
For a standard payload size of \qty{2}{\kibi\byte}, this results at worst in an increase of packet size of 61\%.

\subsection{Computation}
\label{sec:eval:computation}
Like packet size, the computational overhead of \prot{} is dictated by the mix program.
We divide the following evaluation into two parts:
Our first goal is to evaluate the computational overhead of the base packet format.
Then, we want evaluate the computational overhead of individual \prot{} instructions, so that one can gauge the overhead that complex mix programs incur.

We measure package creation time, processing at a intermediate mix node, and at an exit node for \prot{} packets.
Since the mix node requires \emph{some} mix program to process the packet, we use the Sphinx mix program as described in \Cref{sec:programs:sphinx}.
This gives us the opportunity to compare \prot{}'s computational overhead to Sphinx, for which we use the \texttt{java-sphinx}\footnote{\url{https://github.com/rsoultanaev/java-sphinx}, by Robert Soultanaev} implementation.
All measurements were obtained on a machine equipped with an AMD Ryzen 5 5625U  and 16 GiB of RAM, using the \texttt{jmh}\footnote{\url{https://openjdk.org/projects/code-tools/jmh/}} benchmarking framework.
For all measurements, we use a payload size of 1 KiB and a path length of five hops.

Generally, we expect processing times of \prot{} to be higher that those of native Sphinx, due to the additional interpretation of the mix program.
We do not expect large differences, as the computationally heavy cryptographic operations are the same in both formats, and in OmniSphinx, they are represented by a single instruction.

\begin{table}[h]
    \centering
    \begin{tabularx}{0.95\textwidth}{lXXX}\toprule
        Format        & Packet Creation & \multicolumn{2}{c}{Mix Processing}\\
                      &                 & Intermediate Node & Exit Node\\\midrule
        Native Sphinx & \(\qty{1.16}{\ms}\pm\qty{0.00}{\ms}\) & \(\qty{198.29}{\us}\pm\qty{0.32}{\us}\)   & \(\qty{194.17}{\us}\pm\qty{0.28}{\us}\)\\
        \prot{}    & \(\qty{1.12}{\ms}\pm\qty{0.00}{\ms}\)   & \(\qty{283.32}{\us}\pm\qty{1.29}{\us}\)     & \(\qty{280.52}{\us}\pm\qty{1.43}{\us}\)\\\bottomrule
    \end{tabularx}
    \caption{Processing times of \prot{} versus native Sphinx.} 
    \label{tab:computation:base}
\end{table}

\Cref{tab:computation:base} presents our results and confirms the expectations:
Package creation has the same performance for Sphinx and OmniSphinx, as this part is implemented in native Java for both protocols.
For package processing, we see that OmniSphinx is slower by approximately \qty{90}{\us}.
However, the total time needed to process packets is still negligible compared to overheads induced by network latencies.

\begin{table}[h]
    \centering
    \begin{tabularx}{0.95\textwidth}{XXXXXXXXXXXXXXXXXXX}\toprule
        & \rlap{\rotatebox{60}{\instr{Stop}}}
        & \rlap{\rotatebox{60}{\instr{CutBytes}}}
        & \rlap{\rotatebox{60}{\instr{Hash}}}
        & \rlap{\rotatebox{60}{\instr{MAC}}}
        & \rlap{\rotatebox{60}{\instr{Verify}}}
        & \rlap{\rotatebox{60}{\instr{Exponent}}}
        & \rlap{\rotatebox{60}{\instr{Pad}}}
        & \rlap{\rotatebox{60}{\instr{Prg}}}
        & \rlap{\rotatebox{60}{\instr{XOR}}}
        & \rlap{\rotatebox{60}{\instr{Decrypt}}}
        & \rlap{\rotatebox{60}{\instr{Encrypt}}}
        & \rlap{\rotatebox{60}{\instr{Concat}}}
        & \rlap{\rotatebox{60}{\instr{ConcatBytes}}}
        & \rlap{\rotatebox{60}{\instr{Copy}}}
        & \rlap{\rotatebox{60}{\instr{Add}}}
        & \rlap{\rotatebox{60}{\instr{Forward}}}
        & \rlap{\rotatebox{60}{\instr{Load}}}
        & \rlap{\rotatebox{60}{\instr{ForLoop}}}
        \\\midrule
    
        \rotatebox{90}{\textbf{\quad Average time}}
        & \rotatebox{90}{\(\pz{}\pz{}\qty{1.46}{\us}\pm\qty{0.00}{\us}\)}
        & \rotatebox{90}{\(\pz{}\pz{}\qty{1.54}{\us}\pm\qty{0.01}{\us}\)}
        & \rotatebox{90}{\(\pz{}\pz{}\qty{2.19}{\us}\pm\qty{0.01}{\us}\)}
        & \rotatebox{90}{\(\pz{}\pz{}\qty{3.34}{\us}\pm\qty{0.13}{\us}\)}
        & \rotatebox{90}{\(\pz{}\pz{}\qty{1.43}{\us}\pm\qty{0.01}{\us}\)}
        & \rotatebox{90}{\(\qty{153.11}{\us}\pm\qty{0.25}{\us}\)}
        & \rotatebox{90}{\(\pz{}\pz{}\qty{1.54}{\us}\pm\qty{0.01}{\us}\)}
        & \rotatebox{90}{\(\pz{}\pz{}\qty{2.02}{\us}\pm\qty{0.01}{\us}\)}
        & \rotatebox{90}{\(\pz{}\pz{}\qty{1.47}{\us}\pm\qty{0.00}{\us}\)}
        & \rotatebox{90}{\(\pz{}\pz{}\qty{4.13}{\us}\pm\qty{0.02}{\us}\)}
        & \rotatebox{90}{\(\pz{}\pz{}\qty{4.26}{\us}\pm\qty{0.03}{\us}\)}
        & \rotatebox{90}{\(\pz{}\pz{}\qty{1.48}{\us}\pm\qty{0.00}{\us}\)}
        & \rotatebox{90}{\(\pz{}\pz{}\qty{1.49}{\us}\pm\qty{0.00}{\us}\)}
        & \rotatebox{90}{\(\pz{}\pz{}\qty{1.48}{\us}\pm\qty{0.01}{\us}\)}
        & \rotatebox{90}{\(\pz{}\pz{}\qty{1.49}{\us}\pm\qty{0.00}{\us}\)}
        & \rotatebox{90}{\(\pz{}\pz{}\qty{1.48}{\us}\pm\qty{0.00}{\us}\)}
        & \rotatebox{90}{\(\pz{}\pz{}\qty{1.52}{\us}\pm\qty{0.01}{\us}\)}
        & \rotatebox{90}{\(\pz{}\pz{}\qty{1.51}{\us}\pm\qty{0.01}{\us}\)}
        \\\bottomrule
    \end{tabularx}
    \caption{Execution time per \prot{} instruction.}
    \label{tab:computation:instruction}
\end{table}

For the per-instruction cost, we expect all instructions that simply move bytes between registers to be fast.
We expect \instr{MAC}, \instr{Hash}, \instr{Encrypt}, \instr{Decrypt}, \instr{Prg}, and \instr{Exponent} to be slower, as those rely on computationally expensive cryptography.

Our results in \Cref{tab:computation:instruction} confirm our expectations.
All byte-moving operations complete in approximately \qty{1.5}{\us}.
The symmetric primitives \instr{MAC}, \instr{Hash}, \instr{Encrypt}, and \instr{Decrypt} are slower, but only by a factor of 2--3.
The slowest operation is \instr{Exponent}, which represents a public-key cryptographic operation.

\section{Conclusion}
\label{sec:conclusion}
In this work, we investigated if techniques from active networking can provide meaningful benefits in the context of mix networks.
To investigate this question, we implemented \prot{}, a novel mix format, where senders can embed custom mix programs in their packets.
The mix programs are executed by each node on the packet's path and determine how the packet is processed.
With \prot{}, an active mix network can be realized, where mix node operators no longer have to restrict themselves to supporting a single mix format only.
As expected, emulation in \prot{} incurs higher overhead than native execution of the corresponding mix format, but in typical mix network use cases, such as email communication, the total overhead remains manageable:
Sphinx packet processing time increases from approximately \(\qty{200}{\us}\) to approximately \(\qty{300}{\us}\), and headers grow by 33\%.
At the same time, we see concrete benefits of active mix networks over classical ones:
Clients with different needs regarding the mix network's functionality can share the same instance, resulting in better node utilization for operators and a larger choice of nodes for clients.
To further amplify \prot{}'s advantages in practice, future research could determine how auxiliary infrastructure such as directory authorities can be unified across different networks.
In addition to real-world use, active mix networks may become a valuable tool for research in the future, as new mix formats are easy to implement and to deploy without having to rely on node operators or separate testing networks.

\begin{credits}
\subsubsection{\ackname}
We thank the anonymous reviewers for their valuable feedback.

This work has in part been funded by the Helmholtz Association through the KASTEL Security Research Labs (HGF Topic 46.23) and by the German Research Foundation (DFG, Deutsche Forschungsgemeinschaft) as part of Germany’s Excellence Strategy – EXC 2050/2 – Project ID 390696704 – Cluster of Excellence \enquote{Centre for Tactile Internet with Human-in-the-Loop} (CeTI) of Technische Universität Dresden.

\subsubsection{\discintname}
The authors have no competing interests to declare that are
relevant to the content of this article.
\end{credits}
\bibliographystyle{splncs04}
\bibliography{omnisphinx}

@article{chaum1981,
  title = {Untraceable {{Electronic Mail}}, {{Return Addresses}}, and {{Digital Pseudonyms}}},
  author = {Chaum, David L.},
  year = {1981},
  journal = {Commun. ACM},
}

@inproceedings{danezis2009,
  title = {Sphinx: {{A Compact}} and {{Provably Secure Mix Format}}},
  booktitle = {{{IEEE S}}\&{{P}}},
  author = {Danezis, George and Goldberg, Ian},
  year = {2009}
}

@inproceedings{hugenroth2021,
  title = {Rollercoaster: {{An Efficient Group-Multicast Scheme}} for {{Mix Networks}}},
  booktitle = {{{USENIX Security}}},
  author = {Hugenroth, Daniel and Kleppmann, Martin and Beresford, Alastair R.},
  year = {2021}
}

@inproceedings{schadt2024,
  title = {{{PolySphinx}}: {{Extending}} the {{Sphinx Mix Format With Better Multicast Support}}},
  booktitle = {{{IEEE S}}\&{{P}}},
  author = {Schadt, Daniel and Coijanovic, Christoph and Weis, Christiane and Strufe, Thorsten},
  year = {2024},
}

@inproceedings{kuhn2020,
  title = {Breaking and ({{Partially}}) {{Fixing Provably Secure Onion Routing}}},
  booktitle = {{{IEEE S}}\&{{P}}},
  author = {Kuhn, Christiane and Beck, Martin and Strufe, Thorsten},
  year = {2020},
}

@inproceedings{DBLP:conf/wpes/RochetDE24,
  author       = {Florentin Rochet and
                  Jules Dejaeghere and
                  Tariq Elahi},
  title        = {Towards Flexible Anonymous Networks},
  booktitle    = {{WPES}},
  year         = {2024},
}

@inproceedings{DBLP:conf/uss/DingledineMS04,
  author       = {Roger Dingledine and
                  Nick Mathewson and
                  Paul F. Syverson},
  title        = {Tor: The Second-Generation Onion Router},
  booktitle    = {{USENIX} Security},
  year         = {2004},
}

@inproceedings{DBLP:conf/sigcomm/ReiningerAHFHGL21,
  author       = {Michael Reininger and
                  Arushi Arora and
                  Stephen Herwig and
                  Nicholas Francino and
                  Jayson Hurst and
                  Christina Garman and
                  Dave Levin},
  title        = {Bento: safely bringing network function virtualization to {Tor}},
  booktitle    = {{ACM} {SIGCOMM}},
  year         = {2021},
}

@inproceedings{DBLP:conf/ccs/ReiningerAHFGL20,
  author       = {Michael Reininger and
                  Arushi Arora and
                  Stephen Herwig and
                  Nicholas Francino and
                  Christina Garman and
                  Dave Levin},
  title        = {Bento: Bringing Network Function Virtualization to {Tor}},
  booktitle    = {{CCS}},
  year         = {2020},
}

@article{wails2023proteus,
  title={Proteus: Programmable protocols for censorship circumvention},
  author={Wails, Ryan and Jansen, Rob and Johnson, Aaron and Sherr, Micah},
  journal={Free and Open Communications on the Internet},
  year={2023}
}

@inproceedings{rial2025outfox,
  title={Outfox: a Postquantum Packet Format for Layered Mixnets},
  author={Rial, Alfredo and Piotrowska, Ania and Halpin, Harry},
  booktitle={{WPES}},
  year={2025}
}

@misc{cryptoeprint:2024/020,
      author = {Michael Klooß and Andy Rupp and Daniel Schadt and Thorsten Strufe and Christiane Weis},
      title = {{EROR}: Efficient Repliable Onion Routing with Strong Provable Privacy},
      howpublished = {Cryptology {ePrint} Archive, Paper 2024/020},
      year = {2024},
      url = {https://eprint.iacr.org/2024/020}
}

@article{tennenhouse1996towards,
  title={Towards an active network architecture},
  author={Tennenhouse, David L and Wetherall, David J},
  journal={ACM SIGCOMM Computer Communication Review},
  year={1996},
}

@inproceedings{bhattacharjee1997architecture,
  title={An architecture for active networking},
  author={Bhattacharjee, Samrat and Calvert, Kenneth L and Zegura, Ellen W},
  booktitle={International Conference on High Performance Networking},
  year={1997},
}

@article{alexander2002switchware,
  title={The {SwitchWare} active network architecture},
  author={Alexander, D Scott and Arbaugh, William A and Hicks, Michael W and Kakkar, Pankaj and Keromytis, Angelos D and Moore, Jonathan T and Gunter, Carl A and Nettles, Scott M and Smith, Jonathan M},
  journal={IEEE Network},
  year={2002},
}

@article{feamster2014road,
  title={The road to {SDN}: an intellectual history of programmable networks},
  author={Feamster, Nick and Rexford, Jennifer and Zegura, Ellen},
  journal={ACM SIGCOMM Computer Communication Review},
  year={2014},
}

@article{DBLP:journals/cacm/DenningD77,
  author       = {Dorothy E. Denning and
                  Peter J. Denning},
  title        = {Certification of Programs for Secure Information Flow},
  journal      = {Commun. {ACM}},
  year         = {1977},
}

@article{sabelfeld2003language,
  title={Language-based information-flow security},
  author={Sabelfeld, Andrei and Myers, Andrew C},
  journal={IEEE J. Sel. Areas Commun.},
  year={2003},
}

@inproceedings{beato2016,
  title = {Improving the {{Sphinx Mix Network}}},
  booktitle = {Cryptology and {{Network Security}}},
  author = {Beato, Filipe and Halunen, Kimmo and Mennink, Bart},
  year = {2016}
}

@inproceedings {dyer2015,
author = {Kevin P. Dyer and Scott E. Coull and Thomas Shrimpton},
title = {Marionette: A Programmable Network Traffic Obfuscation System},
booktitle = {USENIX Security},
year = {2015},
}

@article{scherer2024a,
  title = {Provable {{Security}} for the {{Onion Routing}} and {{Mix Network Packet Format Sphinx}}},
  author = {Scherer, Philip and Weis, Christiane and Strufe, Thorsten},
  year = {2024},
  journal = {Proc. Priv. Enhancing Technol.},
  urldate = {2025-04-02}
}

@incollection{anderson1996,
  title = {Two Practical and Provably Secure Block Ciphers: {{BEAR}} and {{LION}}},
  booktitle = {Fast {{Software Encryption}}},
  author = {Anderson, Ross and Biham, Eli},
  year = {1996},
}

\appendix

\section{Concrete mix programs}
\label{apx:programs}
In this section, we provide concrete mix programs for the formats described in \Cref{sec:programs}.
Inputs to the program are embedded by the sender as byte string constants.
For Sphinx, we show the relay program in \Cref{apx:pseudo:sphinx-relay} and the exit program in \Cref{apx:pseudo:sphinx-exit}.
For PolySphinx, we show the relay program in \Cref{apx:pseudo:poly-relay}, the replication program in \Cref{apx:pseudo:poly-repl}, and the exit program in \Cref{apx:pseudo:poly-exit}.

\begin{algorithm}
    \caption{Sphinx relay program}
    \label{apx:pseudo:sphinx-relay}
    \inp{Next node \(n\)} \\
    \cmt{Decrypt payload} \\
    \(\instr{ConcatByte}(r_s, \texttt{0x01}, r_s)\) \\
    \(\instr{Hash}(r_s, r_h)\) \\
    \(\instr{Decrypt}(r_h, r_\delta, r_\delta)\) \\
    \cmt{Forward} \\
    \(\instr{Load}(n, r_n)\) \\
    \(\instr{Forward}(r_n)\) \\
    \(\instr{Stop}\)
\end{algorithm}

\begin{algorithm}
    \caption{Sphinx exit program}
    \label{apx:pseudo:sphinx-exit}
    \cmt{Decrypt payload} \\
    \(\instr{ConcatByte}(r_s, \texttt{0x01}, r_s)\) \\
    \(\instr{Hash}(r_s, r_h)\) \\
    \(\instr{Decrypt}(r_h, r_\delta, r_\delta)\) \\
    \cmt{Verify integrity} \\
    \(\instr{CreateZeroes}(\secpar, r_z)\) \\
    \(\instr{CutBytes}(r_\delta, \secpar, r_y)\) \\
    \(\instr{IsEqual}(r_z, r_y)\) \\
    \cmt{Deliver message} \\
    \(\instr{CutBytes}(r_\delta, 2\secpar, r_n)\) \\
    \(\instr{Forward}(r_n)\) \\
    \(\instr{Stop}\)
\end{algorithm}

\begin{algorithm}
    \caption{PolySphinx relay program}
    \label{apx:pseudo:poly-relay}
    \inp{Next node \(n\), node key \(\sigma\)} \\ 
    \(\instr{Load}(\sigma, r_k)\) \\
    \(\instr{Encrypt}(r_k, r_\delta, r_\delta)\) \\
    \(\instr{Load}(n, r_n)\) \\
    \(\instr{Forward}(r_n)\) \\
    \(\instr{Stop}\)
\end{algorithm}

\begin{algorithm}
    \caption{PolySphinx replication program}
    \label{apx:pseudo:poly-repl}
    \inp{Replication factor \(p\), concatenated subheaders \(B\)} \\
    \(\instr{Copy}(r_\delta, r_\Delta)\) \\
    \(\instr{Load}(B, r_B)\) \\
    \cmt{Send \(p\) replicated packets} \\
    \(\instr{ForLoop}(7, p)\) \\
    \null\quad\(\instr{CutBytes}(r_B, 2\kappa, r_n)\) \\
    \null\quad\(\instr{CutBytes}(r_B, \kappa, r_K)\) \\
    \null\quad\(\instr{CutBytes}(r_B, 2\kappa, r_\alpha)\) \\
    \null\quad\(\instr{CutBytes}(r_B, \kappa, r_\gamma)\) \\
    \null\quad\(\instr{CutBytes}(r_B, \tau_\text{Post}, r_\beta)\) \\
    \null\quad\(\instr{Encrypt}(r_K, r_\Delta, r_\delta)\) \\
    \null\quad\(\instr{Forward}(r_n)\) \\
    \(\instr{Stop}\)
\end{algorithm}

\begin{algorithm}
    \caption{PolySphinx exit program}
    \label{apx:pseudo:poly-exit}
    \inp{Replication factor \(p\), key tree seed \(S\), path \(P\), recipient \(R\)} \\
    \(\instr{Load}(S, r_S)\) \\
    \(\instr{Load}(P, r_P)\) \\
    \(\instr{Load}(R, r_R)\) \\
    \cmt{Root of the key tree} \\
    \(\instr{Hash}(r_S, r_\sigma)\) \\
    \(\instr{Hash}(r_\sigma, r_K)\) \\
    \cmt{Build path through the key tree} \\
    \(\instr{ForLoop}(4, p)\) \\
    \null\quad\(\instr{CutBytes}(r_P, 1, r_i)\) \\
    \null\quad\(\instr{Add}(r_\sigma, r_i, r_\sigma)\) \\
    \null\quad\(\instr{Hash}(r_\sigma, r_\sigma)\) \\
    \null\quad\(\instr{Concat}(r_\sigma, r_K, r_K)\) \\
    \cmt{Decrypt the payload} \\
    \(\instr{ForLoop}(2, p + 1)\) \\
    \null\quad\(\instr{CutBytes}(r_K, \kappa, r_\sigma)\) \\
    \null\quad\(\instr{Decrypt}(r_\sigma, r_\delta, r_\delta)\) \\
    \(\instr{Forward}(r_R)\) \\
    \(\instr{Stop}\)
\end{algorithm}

\end{document}